\documentclass[conference]{IEEEtran}
\IEEEoverridecommandlockouts
\usepackage{graphicx}
\usepackage{xcolor}

\usepackage{amsmath}
\usepackage{amssymb}
\usepackage{cite}
\usepackage{cuted}
\newcommand{\linebreakand}{%
  \end{@IEEEauthorhalign}
  \hfill\mbox{}\par
  \mbox{}\hfill\begin{@IEEEauthorhalign}
}    
\usepackage[colorlinks=Flase]{hyperref}
\usepackage{bbm}
\usepackage{algorithm} 
\usepackage{algpseudocode}
\usepackage{caption}
\usepackage{steinmetz}
\usepackage{subcaption}
\usepackage{amsthm}
\usepackage[english]{babel}

\newtheorem{proposition}{Proposition}
\newtheorem{lemma}{Lemma}

\newcounter{relctr} 
\everydisplay\expandafter{\the\everydisplay\setcounter{relctr}{0}} 
\DeclareMathOperator*{\minimize}{minimize} 
\DeclareMathOperator*{\maximize}{maximize} 
\usepackage[figurename=Fig.]{caption}
\newcommand{\norm}[1]{\left\lVert#1\right\rVert}
\usepackage{xpatch}

\newtheoremstyle{remarkstyle}%
  {}
  {}
  {\itshape}
  {}
  {\itshape}
  {.}
  {.5em}
  {}

\theoremstyle{remarkstyle}

\newcommand\labelrel[2]{%
  \begingroup
    \refstepcounter{relctr}%
    \stackrel{\textnormal{(\alph{relctr})}}{\mathstrut{#1}}%
    \originallabel{#2}%
  \endgroup
}
\AtBeginDocument{\let\originallabel\label} 

\ifCLASSINFOpdf
\else
\fi

\begin{document}

\title{Beyond Diagonal Reconfigurable Intelligent Surfaces for Multi-Carrier RF Wireless Power Transfer\\
\thanks{This work is partially supported by the Research Council of Finland (Grants 348515, 362782, and 346208 (6G Flagship)); by the European Commission through the Horizon Europe/JU SNS project Hexa-X-II (Grant no. 101095759); by the Finnish-American Research \& Innovation Accelerator. The work of M. Di Renzo was supported in part by the Nokia Foundation, the French Institute of Finland, and the French Embassy in Finland under the France-Nokia Chair of Excellence in ICT. He was also supported in part by the European Commission through the Horizon Europe project COVER under grant agreement number 101086228, the Horizon Europe project UNITE under grant agreement number 101129618, and the Horizon Europe project INSTINCT under grant agreement number 101139161, as well as by the Agence Nationale de la Recherche (ANR) through the France 2030 project ANR-PEPR Networks of the Future under grant agreement NF-PERSEUS 22-PEFT-004, and by the CHIST-ERA project PASSIONATE under grant agreements CHIST-ERA-22-WAI-04 and ANR-23-CHR4-0003-01. The work of B.~Clerckx is partially supported by UKRI grant EP/Y004086/1, EP/X040569/1, EP/Y037197/1, EP/X04047X/1, and EP/Y037243/1.}
}


\author{
{ Amirhossein~Azarbahram$^{*}$, Onel~L.~A.~L\'{o}pez$^{*}$, Bruno~Clerckx$^{\dagger}$, Marco~Di~Renzo$^{\circ}$, and Matti Latva-aho$^{*}$}
\vspace{2mm}
\\

\small	$^{*}$Centre for Wireless Communications (CWC), University of Oulu, Finland\\
\small	$^{\dagger}$Department of Electrical and Electronic Engineering, Imperial College London, U.K \\
\small $^{\circ}$ Universit\'e Paris-Saclay, CNRS,
CentraleSup\'elec, Laboratoire des Signaux et Syst\`emes, 3 Rue
Joliot-Curie, 91192 Gif-sur-Yvette, France \vspace{2mm}
\\
\small Emails: \{amirhossein.azarbahram, onel.alcarazlopez\}@oulu.fi, b.clerckx@imperial.ac.uk,\\ marco.di-renzo@universite-paris-saclay.fr, and matti.latva-aho@oulu.fi
}


\maketitle

\begin{abstract}

Radio frequency (RF) wireless power transfer (WPT) is promising for promoting sustainability in future wireless systems, but its low end-to-end power transfer efficiency is a critical challenge. For this, reconfigurable intelligent surfaces (RISs) can be leveraged to enhance efficiency by providing nearly passive beamforming gains. Beyond diagonal (BD) RIS is a new RIS variant offering greater performance benefits than traditional diagonal RIS (D-RIS), though its potential for RF-WPT remains unexplored. Motivated by this, we consider a single-input single-output BD-RIS-aided RF-WPT system and we formulate a joint beamforming and waveform optimization problem aiming to maximize the harvested power at the receiver. We propose an optimization framework relying on successive convex approximation, alternating optimization, and semi-definite relaxation. Numerical results show that increasing the number of transmit sub-carriers or RIS elements improves the harvested power. We verify by simulation that BD-RIS leads to the same performance as D-RIS under far-field line-of-sight conditions (in the absence of mutual coupling), while it outperforms D-RIS as the non-line-of-sight components become dominant.

\end{abstract}

\begin{IEEEkeywords}
Radio frequency wireless power transfer (RF-WPT), beyond diagonal reconfigurable intelligent surface (BD-RIS), waveform optimization, nearly passive beamforming.
\end{IEEEkeywords}

\IEEEpeerreviewmaketitle

\section{Introduction}

\IEEEPARstart{R}{adio} frequency (RF) wireless power transfer (WPT) is a promising technology for supporting uninterrupted communications among a myriad of low-power devices in future wireless systems. This can be done by providing multi-user wireless charging capability over large distances while using the wireless communications infrastructure. A key challenge of RF-WPT systems is their inherently low end-to-end power transfer efficiency, which is caused by the power consumption sources at the transmitter side, the channel losses, and the inefficiency at the energy receiver (ER) side. The efficiency of RF-WPT, which is referred to as WPT, can be improved by providing novel solutions for mitigating the adverse impact of these inefficiency sources \cite{lópez2023highpower, ZEDHEXA}.


Waveform optimization and energy beamforming (EB) techniques are promising enablers for efficient WPT. Specifically, optimizing multi-carrier waveforms can enhance efficiency at the ER by leveraging the non-linearity of RF-to-DC conversion \cite{clerckx2018beneficial}. Moreover, EB leads to effectively pointing the transmit signal toward the ER, compensating for wireless channel losses to some extent. EB can be done actively or passively depending on the transmit architecture. An active EB leverages active antenna elements connected to dedicated RF chains, while a passive EB utilizes low-cost nearly passive elements. There exist hybrid architectures, which combine active and passive EB, leading to a reduced number of active elements and a tradeoff between cost/complexity and performance \cite{hybridbeamsurvey}. 


One of the novel enablers of nearly passive EB is reconfigurable intelligent surfaces (RIS), which can improve the performance of wireless systems by providing extra coverage, especially when in the presence of blockages/obstacles. This emerging technology may rely on planar surfaces comprising nearly passive scattering elements, which can introduce amplitude and phase changes to incident electromagnetic waves. Reflective-type RIS can smartly tune the reflected signal and point it toward the desired direction, providing significant EB gains. A recent generalization of diagonal RIS (D-RIS) is given by beyond diagonal RIS (BD-RIS), which is characterized by scattering matrices not constrained to be diagonal, which translate into surfaces where elements/ports are interconnected via tunable impedances \cite{bdris_main_ref}. BD-RIS, through reconfigurable interconnections, provides new degrees of freedom and higher flexibility to manipulate waves by forming group-connected or fully-connected structures \cite{BD-RIS_scattering_clerckx}.

Waveform optimization and EB for WPT systems with non-linear EH have recently attracted considerable research interest, especially when using traditional fully digital transmit architectures, e.g., \cite{BFRFDCsingletone, azarbahram2024deepreinforcementlearningmultiuser, clreckxWFdesign}. Novel low-cost WPT transmit structures for reducing the implementation cost/complexity is another interesting topic,  gaining attention recently. The authors of \cite{azarbahram2024waveform} study the waveform and EB optimization for dynamic metasurface antennas-assisted WPT systems with non-linear EH. Interestingly, the harvested power maximization problem is addressed in \cite{clerckxRISWPT} and \cite{ris_swipt_nonlin} for RIS-aided WPT and simultaneous wireless information and power transfer systems with non-linear EH, respectively. Therein, the authors highlight the extra beamforming gains on the harvested DC power provided by RIS. However, these works only focus on conventional D-RIS, which is limited to a diagonal scattering matrix \cite{BD-RIS_scattering_clerckx}. Although BD-RIS architectures have been widely investigated for communication purposes \cite{BD-RIS_General, BD-RIS_scattering_clerckx, bdris_main_ref, direnzoclerckxmutual}, current research on RIS-aided WPT is still limited to conventional diagonal (local) RIS \cite{ris_swipt_nonlin, clerckxRISWPT}.

All in all, it has been shown that using practical non-linear EH models achieves higher DC power harvesting than linear models \cite{clerckx2018beneficial, clreckxWFdesign}. However, most studies on non-linear EH models address traditional fully digital WPT designs. This paper, in contrast, examines BD-RIS-assisted WPT systems with non-linear EH models, a topic not yet explored in the literature. Our main contributions are: i) we formulate a joint beamforming and waveform optimization problem for a single-input single-output (SISO) multi-carrier WPT system aided by a fully-connected BD-RIS to maximize the DC harvested power; ii) we propose alternating optimization to decouple the waveform optimization and beamforming problems and propose a successive convex approximation (SCA)-based method for waveform optimization, a semi-definite relaxation (SDR)-based method for beamforming, and a randomization-based method to map the obtained solution of SDR into a rank-1 solution; iii) we verify by simulations that in the absence of mutual coupling, BD-RIS achieves the same performance as D-RIS in terms of harvested DC power in far-field line-of-sight (LoS) channels even with multi-carrier signals, while in Rician channels with non-LoS (NLoS) components, BD-RIS provides additional performance gains compared to D-RIS. This extends the findings of \cite{BD-RIS_scattering_clerckx} to multi-carrier WPT.

\textbf{Structure:} Section~\ref{sec:system} covers the system model and problem formulation, while the beamforming and waveform optimization framework is elaborated in Section~\ref{sec:OPTIMIZATION}. Section~\ref{sec:numerical} presents the numerical results and Section~\ref{sec:conclude} concludes the paper.

\textbf{Notations:} Bold lower-case and upper-case letters represent vectors and matrices, respectively. The $\ell_2$-norm operator is denoted by $\norm{\cdot}$. $\Re\{\cdot\}$ and  $\Im\{\cdot\}$ denote the real and imaginary parts of the input. Moreover, $(\cdot)^T$, $(\cdot)^H$, and $(\cdot)^\star$ denote the transpose, transposed conjugate, and conjugate operations, respectively. Additionally, $[\cdot]_{i,l}$ denotes the element in the $l$th column and the $i$th row of a matrix. The vectorization operator is represented by $\mathrm{Vec}(\cdot)$, and its inverse is denoted by ${\mathrm{Vec}^{-1}_{D \times D}(\cdot)}$. Finally, $\mathbf{I}_D$ represents a $D \times D$ identity matrix and $\mathrm{diag}(\mathbf{a})$ refers to a diagonal matrix with its main diagonal being the elements of vector $\mathbf{a}$.

\section{System Model \& Problem Formulation}\label{sec:system}

We consider a multi-carrier SISO WPT system with $N$ sub-carriers aided by a fully-connected BD-RIS with $M$ elements. For simplicity, we assume perfect CSI at the transmitter side as in \cite{clerckxRISWPT, bdris_main_ref, clreckxWFdesign}. Additionally, we assume no direct path exists between the transmitter and receiver, as typical in RIS applications aimed at enhancing coverage in blind spot areas.


\subsection{Transmit and Received Signals}

Multi-carrier waveforms can leverage the rectifier's non-linearity and enhance the performance in terms of DC harvested power \cite{clerckx2018beneficial}. Motivated by this, we consider multiple sub-carriers with $f_n = f_c + (n - 1)\Delta f$ being the frequency of the $n$th sub-carrier, where $f_c$ is the lowest frequency and $\Delta f$ is the sub-carrier spacing. Thus, the transmit signal at time $t$ can be written as $\Re\bigl\{ \sum_{n = 1}^{N} s_n e^{j2\pi f_n t} \bigr\}$, where $s_n$ is the complex weight of the $n$th sub-carrier. This implies unit gain linear signal amplification at the power amplifier.

The transmit signal propagates through the wireless channel. We denote by $\mathbf{h}_{I,n} \in \mathbb{C}^{M\times 1}$ the incident channel between the transmitter and RIS and by $\mathbf{h}_{R,n} \in \mathbb{C}^{M\times 1}$ the channel between the RIS and receiver for the $n$th sub-carrier.
    
The D-RIS is mathematically characterized by diagonal phase-shift matrices. Assuming a multiport network model for the RIS, each RIS element is modeled as a port connected to an independent tunable impedance \cite{univnerini, bdrissportdesign}. This leads to the so-called diagonal scattering matrix, while in BD-RIS the scattering matrix $\mathbf{\Theta}$ is not limited to being diagonal. In the case of reciprocal and lossless BD-RIS\footnote{For the sake of simplicity and since the goal is to investigate the potential gains of BD-RIS for WPT purposes, we consider the case study with no mutual coupling between different ports. The performance of BD-RIS in the presence of electromagnetic mutual coupling has been analyzed in \cite{direnzoclerckxmutual}.}, the scattering matrix is symmetric ($\mathbf{\Theta} = \mathbf{\Theta}^T)$ and unitary ($\mathbf{\Theta}^H\mathbf{\Theta} = \mathbf{I}_M$).

The cascade channel\footnote{The underlying general model and assumptions can be seen in \cite{bdrissportdesign}.} at the $n$th sub-carrier is denoted by $h_n = \mathbf{h}_{R,n}^T\mathbf{\Theta}{\mathbf{h}_{I,n}}$, and the received signal at time $t$ is given by
\begin{equation}\label{eq:rx_signal}
    y(t) = \sum_{n= 1}^{N} \Re\bigl\{ s_n h_n e^{j2\pi f_n t}\bigr\}.
\end{equation}

\subsection{Rectenna}

The ER model, i.e., rectenna, consists of the antenna equivalent circuit and a single-diode rectifier for transforming the RF received signal into harvested DC \cite{azarbahram2024waveform, clerckx2018beneficial, clreckxWFdesign}. By leveraging this model, assuming perfect matching, and using the Taylor expansion, the output current of the rectenna is given by 
\begin{equation}\label{eq:dccurrent}
    i_{dc} =  \sum_{i\ even, i\geq2}^{\bar{n}} K_i \mathbb{E}\bigl\{y(t) ^i\bigr\},
\end{equation}
where $K_2 = 0.17$, $K_4 = 957.25$ \cite{clerckxRISWPT}, and one can consider $\bar{n} = 4$, hence modeling the main sources of non-linearity as a part of the fourth order term, while the second order term represents the ideal linear rectenna model \cite{clreckxWFdesign, clerckxRISWPT}. By expressing \eqref{eq:dccurrent} in the frequency domain and assuming $\bar{n} = 4$, we obtain \cite{clreckxWFdesign}
\begin{align}\label{eq:parseval}
    i_{dc} &= \frac{K_2}{2} \sum_{n}  \norm{s_n h_{n}}^2 + \ldots \nonumber \\ &\frac{3K_4}{8} \sum_{\substack{n_0, n_1, n_2, n_3 \\ n_0 + n_1 = n_2 + n_3}} {({h}_{n_0} {s}_{n_0})}^\star ({h}_{n_1} {s}_{n_1})^\star {({h}_{n_2} {s}_{n_2})} {({h}_{n_3} {s}_{n_3})},
\end{align}
which is more tractable than the sampling-dependent model.

\subsection{Problem Formulation}

The goal of this paper is to maximize the harvested power given a transmit power budget $P_T$. By leveraging the fact that the DC harvested power is an increasing function of the DC current at the ER, the problem can be formulated as
\begin{subequations}\label{graph_problem}
\begin{align}
\label{graph_problem_a} \maximize_{s_n, \mathbf{\Theta}} \quad &  i_{dc} \\
\textrm{subject to} \label{graph_problem_c}  \quad & \mathbf{\Theta} = \mathbf{\Theta}^T, \\
\label{graph_problem_e}  \quad & \boldsymbol{\Theta}^H \boldsymbol{\Theta} = \mathbf{I}_M, \\
\label{graph_problem_e}  \quad & \frac{1}{2} \sum_{n = 1}^{N} \norm{s_n}^2 \leq P_T.
\end{align}
\end{subequations}
Notably, \eqref{graph_problem} is a non-convex problem due to the coupling between the optimization variables and the presence of a quadratic equality constraint, i.e., \eqref{graph_problem_e}. To cope with this, we rely on alternating optimization and decouple the problem into separate problems for optimizing waveform and beamforming.

\section{Beamforming and Waveform Optimization}\label{sec:OPTIMIZATION}

Here, we provide optimization methods for solving the beamforming and waveform-related subproblems. 

\subsection{Waveform Optimization with Fixed $h_n$}

First, we assume the cascade channel $h_n$ is fixed and optimize the signal weights $s_n, \forall n$. We proceed by defining $s_n = \bar{s}_ne^{j\tilde{s}_n}$, where $\bar{s}_n$ and $\tilde{s}_n$ are the amplitude and the phase of $s_n$, respectively. Similarly, we can write $h_n = \bar{h}_ne^{j\tilde{h}_n}$. It is evident that the optimal $\tilde{s}_n$ must compensate for the phases in \eqref{eq:parseval}, leading to a real-valued $i_{dc}$ \cite{clreckxWFdesign}. Thus, it is sufficient to have $\tilde{s}_n^* = - \tilde{h}_n, \forall n$ leading to $\bar{s}_ne^{j\tilde{s}_n^*} h_n = \bar{s}_n\bar{h}_n, \forall n$. Now, the only goal of the optimization is to find the optimal amplitudes for the signal weights at different sub-carriers. For this, the problem for a given $\boldsymbol{\Theta}$ can be reformulated  as
\begin{subequations}\label{fixedPS}
\begin{align}
\label{fixedPS_a} \maximize_{\bar{s}_n} \quad &  {i}_{dc} \\
\textrm{subject to} \label{fixedPS_b}  \quad & \frac{1}{2} \sum_{n = 1}^{N} \bar{s}_n^2 \leq P_T.
\end{align}
\end{subequations}
\begin{lemma}\label{theorem:1}
The DC current in \eqref{eq:parseval} is convex w.r.t. $\bar{s}_n$.
\end{lemma}
\begin{proof}
    Note that ${y(t)}$ is linear w.r.t. $\bar{s}_n$. By leveraging the second-order convexity condition \cite{boyd2004convex} and the linearity of the mathematical average operator, it can be easily verified that $i_{dc}$ is convex w.r.t. $y(t)$. Then, $i_{dc}$ is convex w.r.t. $\bar{s}_n$ since the composition of an affine with a convex function is convex.
\end{proof}
According to Lemma~\ref{theorem:1}, problem \eqref{fixedPS} is not convex since it maximizes a convex function. However, the convexity of \eqref{fixedPS_a} leads to the fact that $\Tilde{i}_{dc}(\bar{s}_n , \bar{s}_n^{(l)}) \leq i_{dc}$, where $\Tilde{i}_{dc}(\bar{s}_n , \bar{s}_n^{(l)})$ is the first order Taylor expansion of $i_{dc}$ at the local point $\bar{s}_n^{(l)}$ formulated as
\begin{align}
    \Tilde{i}_{dc}(\bar{s}_n , \bar{s}_n^{(l)}) &= i_{dc}\big|_{\bar{s}_n = \bar{s}_n^{(l)}} + \sum_{n = 1}^{N} g(\bar{s}_n^{(l)})(\bar{s}_n - \bar{s}_n^{(l)}),
\end{align}
and
\begin{multline}\label{eq:derivativetaylor}
   g(\bar{s}_n) = K_2 \bar{h}_{n}^2{\bar{s}}_n + \frac{3K_4}{2}\biggl[{\bar{h}_n}^4{\bar{s}_n}^3 + 
    2 \sum_{n_1 \neq n} {\bar{h}_n}^2{\bar{h}_{n_1}}^2{\bar{s}_{n_1}}^2\bar{s}_n+ \\
    \sum_{\substack{n_2, n_3 \\ n_2 + n_3 = 2n \\ n_2 \neq n_3}} \bar{h}_{n_2}\bar{h}_{n_3}\bar{h}_{n}^2\bar{s}_{n_2}\bar{s}_{n_3}\bar{s}_{n} + \\
    \sum_{\substack{n_1, n_2, n_3 \\ -n_1 + n_2 + n_3 = n \\ n \neq n_1 \neq n_2 \neq n_3}} \bar{h}_{n_1}\bar{h}_{n_2}\bar{h}_{n_3}\bar{s}_{n_1}\bar{s}_{n_2}\bar{s}_{n_3} \bar{h}_{n}
    \biggr].
\end{multline}
Hereby and by removing the constant terms, the problem can be transformed into a convex problem at the neighborhood of the initial point $\bar{s}_n^{(l)}$, which can be formulated as 
\begin{subequations}\label{fixedPS_approx}
\begin{align}
\label{fixedPS_approx_a} \minimize_{\bar{s}_n} \quad & \xi_1 = -\sum_{n = 1}^{N}  g(\bar{s}_n^{(l)})\bar{s}_n \\
\textrm{subject to} \label{fixedPS_approx_b}  \quad & \frac{1}{2} \sum_{n = 1}^{N} {\bar{s}_n}^2 \leq P_T.
\end{align}
\end{subequations}
Finally, the problem can be iteratively solved using standard convex optimization tools, e.g., CVX \cite{cvxref}.

Algorithm~\ref{alg:waveformsca} illustrates the proposed SCA-based method for obtaining $s_n, \forall n$. First, the scaled match filter approach in \cite{brunolowcomp} is used to initialize the signal weights such that
\begin{equation}\label{eq:SMF}
    s_{n} = e^{-j{{\tilde{h}}_{n}}} {{\bar{h}_{n}}^{\beta}}\sqrt{\frac{2P_T}{\sum_{n_0 = 1}^{N}{\bar{h}_{n_0}}^{2\beta}}}, \quad \forall n.
\end{equation}
Then, the solution is updated iteratively until convergence. By leveraging the lower-bound properties of Taylor-approximation and writing the KKT conditions of \eqref{fixedPS}, it can be seen that \eqref{fixedPS_approx_a} is monotonically increasing and the obtained solution by Algorithm~\ref{alg:waveformsca} satisfies the KKT conditions of \eqref{fixedPS}.

\begin{algorithm}[t]
	\caption{SCA-based waveform optimization (SCA-WF).} \label{alg:waveformsca}
	\begin{algorithmic}[1]
            \State \textbf{Input:} $\{h_n\}_{\forall n}$, $\rho_s$, $\upsilon$ \quad \textbf{Output:} $s_n^{(l)}$
             \State \textbf{Initialize:} Initialize $s_n^{(l)}, \forall n$ using \eqref{eq:SMF}, $\xi_1 =\infty$
            \Repeat
                \State \hspace{-2mm} $\xi_1^\star \leftarrow \xi_1$
                \State \hspace{-2mm} Solve \eqref{fixedPS_approx} to obtain $\bar{s}_n^{(l + 1)}, \forall n$
                \State \hspace{-2mm} $s_n^{(l + 1)} \leftarrow \bar{s}_n^{(l + 1)}e^{-j\tilde{h}_n}, \forall n$ 
                \State \hspace{-2mm} Compute $\xi_1$ using \eqref{fixedPS_approx_a}, $s_n^{(l)} \leftarrow s_n^{(l + 1)}$, $l \leftarrow l + 1$
            \Until{$\norm{1 - {\xi_1^\star}/{\xi_1}}\leq \upsilon$}
                
\end{algorithmic} 
\end{algorithm}

\subsection{SDR-based Beamforming with Fixed $s_n$}\label{sec:SDROPTIMIZATION}

Herein, we provide an optimization method for the beamforming problem given $s_n, \forall n$, which adapts the approach and mathematical reformulations provided in \cite{clerckxRISWPT}. Let us proceed by rewriting the optimization problem for fixed $s_n$ as
\begin{subequations}\label{scatter_probelm}
\begin{align}
\label{scatter_probelm_a} \maximize_{\boldsymbol{\Theta}} \quad &  i_{dc} \\
\textrm{subject to} \label{scatter_probelm_b}  \quad & \boldsymbol{\Theta} = \boldsymbol{\Theta}^T, \\
\label{scatter_probelm_c}  \quad & \boldsymbol{\Theta}^H \boldsymbol{\Theta} = \mathbf{I}_M,
\end{align}
\end{subequations}
which is highly complicated and non-convex due to the unitary constraint and the $i_{dc}$ non-linearity.

\begin{proposition}\label{theorem:vecotorization2}
    The cascade channel, i.e., $h_n = \mathbf{h}_{R,n}^T\mathbf{\Theta}{\mathbf{h}_{I,n}}$, can be rewritten as 
    \begin{equation}\label{eq:hnreform2}
         h_n = \mathbf{a}_n^T \boldsymbol{\theta},
    \end{equation}
    where $\mathbf{a}_n =  \mathbf{P}^T\mathrm{Vec}(\mathbf{h}_{I,n}\mathbf{h}_{R,n}^T) \in \mathbb{C}^{M(M + 1)/2 \times 1}$. Moreover, $\boldsymbol{\theta} \in \mathbb{C}^{M(M + 1)/2 \times 1}$ is the vector containing the lower/upper-triangle elements in $\mathbf{\Theta}$ and $\mathbf{P} \in \{0, 1\}^{M^2\times M(M + 1)/2}$ is a permutation matrix such that
\begin{equation}\label{eq:permute_matrix} [\mathbf{P}]_{M(m - 1) + n, k} = 
    \begin{cases}
        1, & k = {m(m - 1)}/{2} + n,\  1 \leq n \leq m 
        \\
        1, & k = {n(n - 1)}/{2} + m,\  m < n \leq M \\
        0, & \text{otherwise}.
    \end{cases}
\end{equation} 
\end{proposition}
\begin{proof}
We proceed by rewriting $h_n$ as 
\begin{align}
    h_n &= \mathbf{h}_{R,n}^T\boldsymbol{\Theta}{\mathbf{h}_{I,n}} 
    = \mathrm{Tr}(\mathbf{h}_{R,n}^T\boldsymbol{\Theta}{\mathbf{h}_{I,n}})  \nonumber \\
    &\labelrel={myinlab:1} \mathrm{Tr}({\mathbf{h}_{I,n}}\mathbf{h}_{R,n}^T\boldsymbol{\Theta})  = \mathrm{Tr}(\mathbf{H}_n\boldsymbol{\Theta}) \nonumber \\
    &\labelrel={myinlab:2} {\mathrm{Vec}(\mathbf{H}_n)}^T\mathrm{Vec}(\boldsymbol{\Theta}) 
    \labelrel={myinlab:3}{\mathrm{Vec}(\mathbf{H}_n)}^T\mathbf{P}\boldsymbol{\theta}
    = \mathbf{a}_n^T \boldsymbol{\theta}.
\end{align}
Assume that $\mathbf{D}$, $\mathbf{F}$, and $\mathbf{H}$ are arbitrary matrices. Hereby, \eqref{myinlab:1} comes from $\mathrm{Tr}(\mathbf{DFH}) = \mathrm{Tr}(\mathbf{HDF})$ and \eqref{myinlab:2} comes from $\mathrm{Tr}(\mathbf{D}^T\mathbf{F}) = \mathrm{Vec}(\mathbf{D})^T\mathrm{Vec}(\mathbf{F})$. Moreover, by defining $\boldsymbol{\theta} \in \mathbb{C}^{M(M + 1)/2 \times 1}$ as the vector containing the lower/upper-triangle elements in $\mathbf{\Theta}$, one can design a permutation matrix $\mathbf{P} \in \{0, 1\}^{M^2\times M(M + 1)/2}$ such that $\mathbf{P}\boldsymbol{\omega} = \mathrm{Vec}(\mathbf{\Omega})$ holds, which leads to \eqref{myinlab:3}.
\end{proof}

Leveraging Proposition~\ref{theorem:vecotorization2} allows to remove the constraint \eqref{scatter_probelm_b} by using $\boldsymbol{\theta}$ as the optimization variable. However, the complexity caused by the unitary constraint still remains. For this, we leverage the idea in \cite{clerckxRISWPT} and define $\mathbf{z}_n = s_n\mathbf{a}_n$, $\mathbf{D}_0 = \mathbf{z}_1\mathbf{z}_1^H + \ldots + \mathbf{z}_N\mathbf{z}_N^H$, $\mathbf{D}_1 = \mathbf{z}_1\mathbf{z}_2^H + \ldots + \mathbf{z}_{N-1}\mathbf{z}_N^H$, and $\mathbf{D}_{N - 1} = \mathbf{z}_1\mathbf{z}_N^H$. Hereby, $i_{dc}$ can be reformulated as 
\begin{multline}\label{eq:idc_reform1}
    i_{dc} = \frac{1}{2} K_2 \boldsymbol{\theta}^H \mathbf{D}_0 \boldsymbol{\theta} + \frac{3}{8} K_4 \boldsymbol{\theta}^H \mathbf{D}_0 \boldsymbol{\theta} (\boldsymbol{\theta}^H \mathbf{D}_0 \boldsymbol{\theta})^H  \\ + \frac{3}{4} K_4 \sum_{n = 1}^{N - 1} \boldsymbol{\theta}^H \mathbf{D}_n \boldsymbol{\theta} (\boldsymbol{\theta}^H \mathbf{D}_n \boldsymbol{\theta})^H.
\end{multline}
Next, we need to formulate the constraint \eqref{scatter_probelm_c} as a function of the new optimization variable $\boldsymbol{\theta}$. For this, let us proceed by defining a permutation matrix $\mathbf{P}_i$, 
which extracts the $i$th row of $\boldsymbol{\Theta}$ from $\boldsymbol{\theta}$. Hereby, \eqref{scatter_probelm_c} can be rewritten as 
\begin{align}\label{eq:permutetheta}
    (\mathbf{P}_i \boldsymbol{\theta})^H(\mathbf{P}_j \boldsymbol{\theta}) = \mathrm{Tr}(\boldsymbol{\theta}\boldsymbol{\theta}^H \bar{\mathbf{P}}_{i, j}) =  \begin{cases}
        1, & i = j, 
        \\
        0, & i \neq j,
    \end{cases}
\end{align}
where $\bar{\mathbf{P}}_{i, j} = \mathbf{P}_i^T \mathbf{P}_j$ and $\mathbf{P}_i$ is a permutation matrix containing the $(iM - M + 1)$th to the $iM$th row of $\mathbf{P}$.

Now, we define $d_n = \boldsymbol{\theta}^H \mathbf{D}_n \boldsymbol{\theta}$, $\mathbf{d} = [ d_1, \ldots, d_N]^T$, and positive semidefinite matrices $\mathbf{K}_0 = \mathrm{diag}\{\frac{3}{8} K_4, \frac{3}{4} K_4, \ldots, \frac{3}{4} K_4\} \succeq 0$ and $\mathbf{X} = \boldsymbol{\theta} \boldsymbol{\theta}^H$. Hereby, the problem can be reformulated as 
\begin{subequations}\label{SDP_problem_gen}
\begin{align}
\label{SDP_problem_gen_a} \maximize_{\mathbf{d}, \mathbf{X} \succeq 0} \quad &  \frac{1}{2} K_2 d_0 +  \mathbf{d}^H \mathbf{K}_0 \mathbf{d} \\
\textrm{subject to} \label{SDP_problem_gen_c}  \quad & \mathrm{Tr}(\mathbf{X} \bar{\mathbf{P}}_{i, j}) = 1,\ \forall i = j, \\
\label{SDP_problem_gen_f}  \quad & \mathrm{Tr}(\mathbf{X} \bar{\mathbf{P}}_{i, j}) = 0,\ \forall i \neq j, \\
\label{SDP_problem_gen_d}  \quad & d_n = \mathrm{Tr}(\mathbf{X} \mathbf{D}_n),\ \forall n, \\
\label{SDP_problem_gen_e}  \quad & \mathrm{rank}(\mathbf{X}) = 1.
\end{align}
\end{subequations}

Note that problem \eqref{SDP_problem_gen} is still non-convex and challenging to solve since it deals with the maximization of a convex objective function and includes a rank-1 constraint. For \eqref{SDP_problem_gen_a}, SCA can be used to iteratively update the objective function by approximating it using its first-order Taylor expansion. Specifically, the quadratic term in \eqref{SDP_problem_gen_a} can be approximated in the neighborhood of $\mathbf{d}^{(l)}$ by \cite{clerckxRISWPT}
\begin{equation}\label{eq:taylor_sdp}
    f(\mathbf{d}, \mathbf{d}^{(l)}) = 2 \Re \bigl\{ {\mathbf{d}^{(l)}}^H \mathbf{K}_0 \mathbf{d} \bigr\} - {\mathbf{d}^{(l)}}^H \mathbf{K}_0 {\mathbf{d}^{(l)}}. 
\end{equation}

Since $f(\mathbf{d}, \mathbf{d}^{(l)}) \leq \mathbf{d}^H \mathbf{K}_0 \mathbf{d}$ always holds, it can be used as a lower bound and maximizing the Taylor approximation iteratively leads to maximizing the original quadratic term. We use SDR to relax the rank-1 constraint and reformulate the problem in the neighborhood of the local point $\mathbf{d}^{(l)}$ as
\begin{subequations}\label{SDP_problem_rank}
\begin{align}
\label{SDP_problem_rank_a} \minimize_{\mathbf{X}} \quad & \Omega = \mathrm{Tr}(\mathbf{K}_1\mathbf{X}) \\
\textrm{subject to} \label{SDP_problem_rank_e}  \quad & \mathbf{X} \succeq 0, \\
\quad & \eqref{SDP_problem_gen_c}, \eqref{SDP_problem_gen_f}, \nonumber
\end{align}
\end{subequations}
where $\mathbf{K}_1 = \mathbf{J} + \mathbf{J}^H$ and
\begin{equation}
    \mathbf{J} = -\frac{K_2}{4} \mathbf{D}_0 - \frac{3K_4}{8} d_0^{(l)} \mathbf{D}_0 - \frac{3 K_4}{4} \sum_{n = 1}^{N - 1} {d_n^{(l)}} \mathbf{D}_n.
\end{equation}
This problem is a standard semi-definite programming (SDP) problem that can be solved using convex optimization tools. Moreover, if the obtained $\mathbf{X}^*$ is a rank-1 matrix, the SDR is tight and $\mathbf{X}^*$ is  a stationary point of problem \eqref{SDP_problem_gen}, leading to a local optimum point extracted by $\mathbf{X}^* = \boldsymbol{\theta}^* {\boldsymbol{\theta}^*}^H$. However, it might happen that $\mathrm{rank}(\mathbf{X}^*) > 1$, which leads to $\mathbf{X}^*$ satisfying the KKT conditions of problem \eqref{SDP_problem_rank} (see \cite{clerckxRISWPT} for the proof). For this, we obtain an approximate $\boldsymbol{\theta}^*$ using the Gaussian randomization method in \cite{guss_randomization}. In this scenario, constructing $\mathbf{\Theta}$ with $\boldsymbol{\theta}^*$ results in a symmetric $\mathbf{\Theta}$, though it may not necessarily be a unitary matrix. Therefore, it is crucial to project the final solution into the problem's feasible space. Let us proceed by writing $\mathbf{\Theta}' = \mathrm{Vec}^{-1}(\mathbf{P}\boldsymbol{\theta}^*)$. Then, by leveraging the fact that $\mathbf{\Theta}'$ is symmetric, we can write the singular value decomposition (SVD) as $\mathbf{\Theta}' = \mathbf{Q}\boldsymbol{\Sigma}\mathbf{Q}^T$, where $\mathbf{Q}$ is a unitary matrix and $\boldsymbol{\Sigma}$ is a diagonal matrix containing the singular values of $\mathbf{\Theta}'$. It is evident that if the diagonal elements of $\boldsymbol{\Sigma}$ are unit modulus,  $\mathbf{\Theta}'$ is unitary. However, this only happens when $\mathrm{rank}(\mathbf{X}^*) = 1$, while for higher-rank cases, we propose a randomization-based method to construct a $\boldsymbol{\Sigma}$ with unit modulus diagonal elements.

\begin{algorithm}[t]
	\caption{SDR-based passive beamforming and waveform optimization for fully-connected BD-RIS (SDR-BDRIS).} \label{alg:SDPRIS}
	\begin{algorithmic}[1]
            \State \textbf{Input:} $\upsilon$, $\beta$, $\mathbf{h}_{R,n}, \mathbf{h}_{I,n}, \forall n$ \quad \textbf{Output:} $s_n^{(l)}, \forall n$, $\mathbf{\Theta}^*$
            \State \textbf{Initialize:}  Choose $\mathbf{\Theta}^{(l)}$ randomly,  $f^\star =0$, $i_{dc} = \infty$
            \State Compute $\mathbf{P}$ and $h_n$ using \eqref{eq:permute_matrix} and \eqref{eq:hnreform2}
            \State Initialize $s_n^{(l)}, \forall n$ using \eqref{eq:SMF}
            \State Compute $d_n^{(l)} = \boldsymbol{\theta}^{(l)}\mathbf{D}_n{\boldsymbol{\theta}^{(l)}}^H$, where $\boldsymbol{\theta}^{(l)} = \mathbf{P}^{-1}\mathrm{Vec}(\mathbf{\Theta}^{(l)})$
            \Repeat\label{algBDRIS:line:alter_start}
                \State \hspace{-2mm} $\Omega  =\infty$, $i_{dc}^\star \leftarrow i_{dc}$
                \Repeat\label{algBDRIS:passstart}
                    \State \hspace{-2mm} $\Omega ^\star \leftarrow \Omega $, solve \eqref{SDP_problem_rank} to obtain $\mathbf{X}$
                    \State \hspace{-2mm} Compute $\boldsymbol{\theta}^{(l)}$ using Gaussian randomization
                    \State \hspace{-2mm} $d_n^{(l)} = \boldsymbol{\theta}^{(l)}\mathbf{D}_n{\boldsymbol{\theta}^{(l)}}^H$, $\mathbf{\Theta}^{(l)} = \mathrm{Vec}^{-1}(\mathbf{P}\boldsymbol{\theta}^{(l)})$
                    \State \hspace{-2mm} $l \leftarrow l + 1$
                \Until{$\norm{1 - {\Omega ^\star}/{\Omega }}\leq \upsilon$}\label{algBDRIS:passend}
                \State \hspace{-2mm} $h_n = \mathbf{h}_{R,n}^T\mathbf{\Theta}^{(l)}{\mathbf{h}_{I,n}}$ and run Algorithm~\ref{alg:waveformsca} to update $s_n^{(l)}$
                \State \hspace{-2mm} Compute $i_{dc}$ using \eqref{eq:parseval}
            \Until{$\norm{1 - i_{dc}^\star/ i_{dc}}\leq \upsilon$}\label{algBDRIS:line:alter_end}
            \State Run Algorithm~\ref{alg:PS_obtain} to obtain $\mathbf{\Theta}^*$
\end{algorithmic} 
\end{algorithm}

\begin{algorithm}[t]
	\caption{Randomization-based method for obtaining a feasible $\mathbf{\Theta}$.} \label{alg:PS_obtain}
	\begin{algorithmic}[1]
            \State \textbf{Input:} $\mathbf{\Theta}$, $K$, $\mathbf{h}_{R,n}, \mathbf{h}_{I,n}, s_n, \forall n$ \quad \textbf{Output:} $\mathbf{\Theta}^*$
             \State \textbf{Initialize:} 
             \State  Compute the SVD of $\mathbf{\Theta}$ to obtain $\mathbf{\Theta} = \mathbf{Q}\boldsymbol{\Sigma}\mathbf{Q}^T$, $i^\star_{dc} = 0$
            \For{$k = 1, \ldots, K$}
                \State \hspace{-2mm} Generate random $\phi_i \in [0, 2\pi], \forall i$ 
                \State \hspace{-2mm} Set $\boldsymbol{\phi} = [e^{j\phi_1}, \ldots, e^{j\phi_M}]^T$ and $\boldsymbol{\Sigma}' = \mathrm{diag}(\boldsymbol{\phi})$
                \State \hspace{-2mm} Compute $\mathbf{\Theta} = \mathbf{Q}\boldsymbol{\Sigma}'\mathbf{Q}^T$ and $i_{dc}$ using \eqref{eq:parseval}
                \If{$i_{dc} > i^\star_{dc}$}
                    \State $i^\star_{dc} \leftarrow i_{dc}$, $\mathbf{\Theta}^* \leftarrow \mathbf{\Theta}$
                \EndIf
            \EndFor
\end{algorithmic} 
\end{algorithm}

Algorithm~\ref{alg:SDPRIS} describes the proposed SDR-based method for passive beamforming and waveform optimization for fully-connected BD-RIS. First, the optimization variables are initialized. Then, the waveform and scattering matrix are optimized in an alternative fashion through lines \ref{algBDRIS:line:alter_start}-\ref{algBDRIS:line:alter_end}. Specifically, beamforming is done by solving \eqref{SDP_problem_rank} iteratively in lines \ref{algBDRIS:passstart}-\ref{algBDRIS:passend}, followed by iterative waveform optimization using Algorithm~\ref{alg:waveformsca}. After that, Algorithm~\ref{alg:PS_obtain} is utilized to construct a feasible solution $\mathbf{\Theta}^*$ based on the characteristics of the obtained solution $\mathbf{\Theta}$. First, the SVD of $\mathbf{\Theta}$ is computed to obtain a unitary matrix $\mathbf{Q}$. Then, random phase shifts are generated for $K$ iterations to construct new $\boldsymbol{\Sigma}'$ matrices and their corresponding feasible solution $\mathbf{\Theta}$. Finally, the constructed $\mathbf{\Theta}$ with the best $i_{dc}$ is selected as the final solution.



\subsection{Complexity Analysis}\label{subsec:companal}

Herein, we discuss the time complexity of the proposed algorithms in details:\\
     \textit{SCA-WF:} Algorithm~\ref{alg:waveformsca} requires solving a quadratic program \cite{boyd2004convex} in each SCA iteration. Notably, the complexity of quadratic programs scales with a polynomial function of the problem size, while the degree of the polynomial mainly depends on the type of solver. Let us consider a simple solver based on the Newton method, which has $\mathcal{O}(n^3)$ complexity \cite{boyd2004convex}, where $n$ is the problem size and the number of variables is $N$, leading to $n$ scaling with $N$. \\
     \textit{SDR-BDRIS:} The number of variables in \eqref{SDP_problem_rank} is ${\bar{M}(\bar{M} + 1)}/2$ with $\bar{M} = M(M + 1)/2$, while the rest of the entries in $\mathbf{X}$ are determined according to the Hermitian structure, and the sizes of these Hermitian matrix sub-space is ${\bar{M}}^2$. Additionally, the number of constraints scales with ${M}$ in \eqref{SDP_problem_rank}. It is shown that for a given accuracy, the complexity of SDP problems grows at most with $\mathcal{O}(n^{1/2})$, where $n$ is the problem size, scaling with the number of constraints and variables \cite{semidef-boyd}.

\section{Numerical Analysis}\label{sec:numerical}

In this section, we evaluate the system's performance in a WiFi-like scenario at a carrier frequency of $f_c = 2.4$~GHz. We consider the path loss due to large-scale fading at a distance $d$ to be $L_0d^{-\kappa}$, where  $L_0= 40$~dB is the path loss at a reference distance of 1~m and $\kappa = 2$ is the path loss exponent. Moreover, the incident and reflective path are both considered to be 2~m.

The transmit power is $P_T = 50$~dBm and the channels are modeled with quasi-static Rician fading given by
\begin{equation}
    \mathbf{h}_n = \sqrt{\kappa/(\kappa + 1)}\mathbf{h}_n^{\text{LoS}} + \sqrt{1/(\kappa + 1)}\mathbf{h}_n^{\text{NLoS}},
\end{equation}
where $\kappa$ is the Rician factor. The NLoS part is modeled with Rayleigh fading considering $L = 18$ delay taps with realizations following a circularly symmetric complex Gaussian distribution with a random power $p_l$, such that $\sum_{l = 1}^L p_l = 1$.

\begin{figure}[t]
    \centering
    \includegraphics[width=0.6\columnwidth]{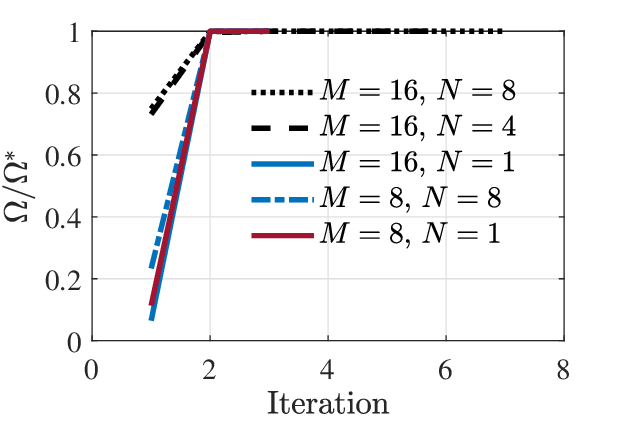}
    \caption{The convergence performance of SDR-BDRIS with a random channel realization and different $M$ and $N$.}
    \label{fig:convergence}
\end{figure}

Fig.~\ref{fig:convergence} illustrates the convergence performance of the proposed SDR-BDRIS. It is seen that SDR-BDRIS iteratively converges toward a local optimum solution. Note that in this figure, the number of required iterations for convergence in different setups is the same. However, the time complexity of solving the SDP problem in \eqref{SDP_problem_rank} drastically increases with $M$ since the hermitian matrix sub-space is in the order of $M^4$.

\begin{figure}[t]
    \centering
    \includegraphics[width=0.49\columnwidth]{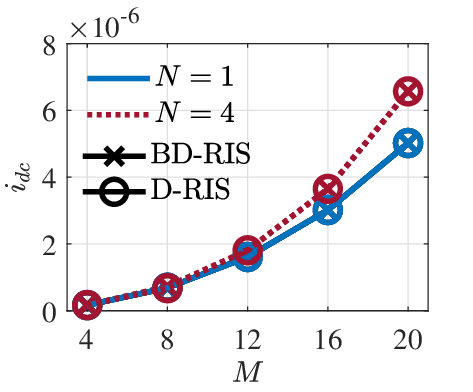}
    \includegraphics[width=0.49\columnwidth]{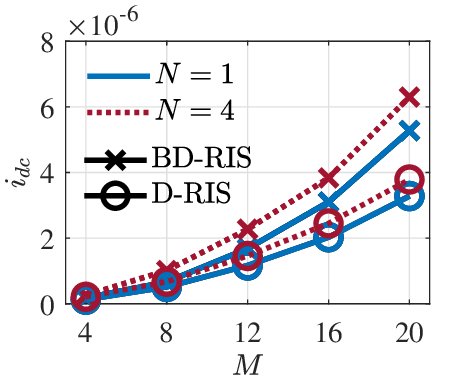} 
    \caption{The average $i_{dc}$ at the ER for (a) LoS channel (left) and (b) Rician channel with $\kappa = 0$~dB (right) as a function of $M$ for $N = 1$ and $N = 4$.}
    \label{fig:DBDoverM}
\end{figure}

\begin{figure}[t]
    \centering
    \includegraphics[width=0.49\columnwidth]{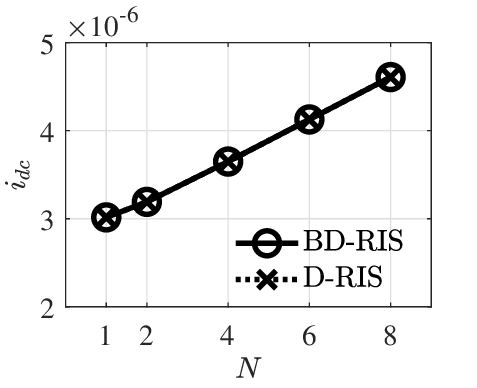} 
    \includegraphics[width=0.49\columnwidth]{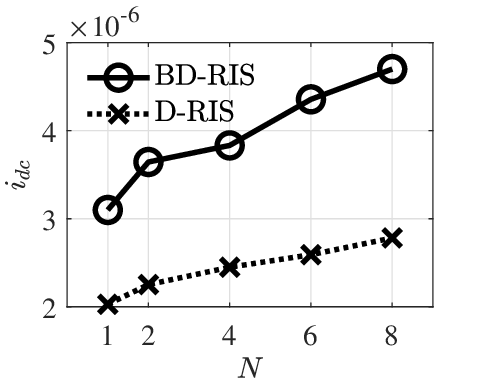}
    \caption{The average $i_{dc}$ at the ER for (a) LoS channel (left) and (b) Rician channel with $\kappa = 0$~dB (right) as a function of $N$ for $M = 16$.}
    \label{fig:DBDoverN}
\end{figure}

Fig.~\ref{fig:DBDoverM} shows the average $i_{dc}$ as a function of $M$ for D-RIS and BD-RIS. The D-RIS results are obtained using the approach proposed in \cite{clerckxRISWPT}. As expected, it is observed that increasing the number of elements increases $i_{dc}$. Moreover, Fig.~\ref{fig:DBDoverM}.a illustrates that BD-RIS achieves the same performance as D-RIS for both single-carrier and multi-carrier systems in pure LoS conditions. However, when $\kappa$ decreases and the channel tends to become frequency-selective, as in Fig.~\ref{fig:DBDoverM}.b, BD-RIS can leverage the extra degrees of freedom to impact different components of the channel effectively such that $i_{dc}$ becomes higher compared to D-RIS. 

Fig.~\ref{fig:DBDoverN} presents the average $i_{dc}$ as a function of $N$ for D-RIS and BD-RIS. It is shown that for LoS channels, BD-RIS and D-RIS achieve the same performance for any number of sub-carriers, while for the Rician channel with NLoS components, the BD-RIS becomes the favorable option and the performance gap increases with $N$. Moreover, it is seen that increasing the number of sub-carriers improves the performance.

\section{Conclusion \& Future Work}\label{sec:conclude}

In this paper, we considered a BD-RIS-aided SISO WPT system with EH non-linearity. Moreover, we formulated a joint beamforming and waveform optimization problem to maximize the harvested power at the ER. We proposed the SDR-BDRIS approach relying on alternating optimization, SCA, and SDR to solve the problem. The simulation results proved that BD-RIS achieves the same performance as D-RIS under far-field LoS conditions (in the absence of mutual coupling), while it outperforms D-RIS in NLoS cases. As expected, our findings demonstrated that increasing the number of elements or sub-carriers improves performance.



\bibliographystyle{ieeetr}
\bibliography{ref_abbv}

\end{document}